\documentclass[A4]{article}
\usepackage[margin=2cm]{geometry}
\usepackage{amsmath}
\usepackage{amssymb}
\usepackage{amsfonts}
\usepackage{amsthm}
\usepackage{graphicx,color}
\usepackage{float}
\usepackage{rotating}
\usepackage{tikz}
\usepackage{accents} 
\usepackage{todonotes} 
\usepackage{subcaption}

\newcounter{NN}
\newtheorem{proposition}[NN]{Proposition}

\newtheorem{lemma}[NN]{Lemma}

\newcommand{\wh}{\widehat}
\newcommand{\wt}{\widetilde}
\newcommand{\ol}{\overline}
\newcommand{\utl}{\underaccent{\tilde}}
\newcommand{\uht}{\underaccent{\hat}}
\newcommand{\udo}{\underaccent{\bar}}

\newcommand{\vf}{\varphi}

\numberwithin{equation}{section}

\def\P{P}
\def\Q{Q}

\def\w{{Y}}
\def\V{{V}}
\def\W{{W}}

\def\b#1{{\bf #1}}

\begin{document}
\title{On trilinear and quadrilinear equations associated with the lattice Gel'fand-Dikii hierarchy}
\author{P.H.~van der Kamp$^{1,3}$, F.W. Nijhoff$^2$,  D.I. McLaren$^{1}$, G.R.W Quispel$^{1}$}
\date{
$^1$Department of Mathematical and Physical Sciences\\
La Trobe University, Victoria 3086, Australia\\
$^2$School of Mathematics,
University of Leeds, UK\\
[5mm]
$^3$ Corresponding author, P.vanderKamp@LaTrobe.edu.au
}
\maketitle

\begin{abstract}
Introduced in \cite{ZZN}, the trilinear Boussinesq equation is the natural form of the equation for the 
$\tau$-function of the lattice Boussinesq system. In this paper we study various aspects of this equation: 
its highly nontrivial derivation from the bilinear lattice AKP equation under dimensional reduction, a quadrilinear dual lattice equation, conservation laws, and 
periodic reductions leading to higher-dimensional integrable maps and their Laurent property.
Furthermore, we consider a higher Gel'fand-Dikii lattice system, its periodic reductions and Laurent property.
As a special application, from both a trilinear Boussinesq recurrence as well as a higher Gel'fand-Dikii  system of three bilinear recurrences, we establish Somos-like integer sequences. \\ 

\vspace{.2cm} 
\hfill{\textit{To celebrate 120 years from the day of birth of A. Kolmogorov}}

\end{abstract}

\begin{center}{\small
}
\end{center}

\section{Introduction}
Lattice versions of the Boussinesq (BSQ) equations were first introduced in \cite{NPCQ}. They can be written as discrete equations on a regular 2D lattice with a 9-point stencil, or alternatively 
as multi-component quad equations, cf. \cite{Nij,NijTong}, cf. also \cite{Walker}. Some of the 
3-component systems given in \cite{Walker} were generalised in \cite{Hiet}, leading to some 
interesting parameter extensions, referred to here as 'extended BSQ systems', cf. also \cite{HJN}. 
In \cite{ZZN} the various classes of extended systems from \cite{Hiet} were identified 
within the framework of 'direct linearization' scheme, which allows 
one not only to find the interrelations between the various BSQ systems, but also to construct their Lax pairs 
and explicit solutions. Alongside, a novel trilinear equation was found 
in \cite{ZZN} involving the BSQ $\tau$-function. It is the latter equation that forms the focus 
of attention in the present paper. 

This trilinear Boussinesq equation is defined on a 9-point stencil and adopts the most general form
\begin{equation} \label{eq:TB}
\big( A\wh{\tau}\uht{\tau}+B\wt{\tau} \utl{\tau}+C\wh{\wt{\tau}}\uht{\utl{\tau}} \big)\tau + D\big(\wh{\utl{\tau}} \wt{\tau} \uht{\tau}+ \wt{\uht{\tau}} \wh{\tau} \utl{\tau} \big)
=0,
\end{equation}
for a function $\tau=\tau(n,m)$ of discrete variables $n,m$, and where 
the $\wt{\phantom{a}}$ and $\wh{\phantom{a}}$ denote elementary shifts on the lattice labelled by these discrete 
variables, i.e. $\wt{\tau}=\tau(n+1,m)$, $\wh{\tau}=\tau(n,m+1)$, while the underaccents denote the corresponding 
reverse shifts, i.e. $\utl{\tau}=\tau(n-1,m)$, $\uht{\tau}=\tau(n,m-1)$. 
In \cite[Equation (71)]{ZZN} the parameters $A,B,C,D$ were identified as 
$A=3p^2-2\alpha_2 p+\alpha_1$, $B=3q^2-2\alpha_2 q+\alpha_1$, $C=-(p-q)^2$ and 
$D=\alpha_2(p+q)-(p^2+pq+q^2)-\alpha_1$, in which case the solution structure was exhibited in the context of the direct linearisation approach, which allows the construction of explicit soliton and inverse scattering type solutions.  Elliptic solutions, for an elliptic parametrisation of  coefficients \cite[Equation (5.29)]{NSZ}, were given recently in \cite{NSZ}. 

In this paper, we will study further the trilinear BSQ \eqref{eq:TB} with general coefficients. 
In section 2, we show how it is obtained as a dimensional reduction from the bilinear AKP equation (also known 
as Hirota-Miwa equation), which unlike the reduction from AKP to the bilinear Hirota form of the lattice Korteweg-de Vries (KdV) equation is quite subtle. However, it provides a window on how to perform higher-order reductions to multilinear equations associated with the lattice Gel'fand-Dikii (GD) hierarchy of \cite{NPCQ}. 
In section 3, we derive conservation laws from three of the four conservation laws for the AKP equation that were given in \cite{D2}. The characteristics of these conservation laws, give rise to a quadrilinear dual equation, cf. \cite{D1,D2}, which corresponds to a higher 
analogue of the discrete time Toda (HADT) equation, which was derived in the framework of orthogonal polynomials in two variables related via an elliptic curve, cf. \cite{SNV}. In section 4, we show how to construct initial value problems for the trilinear BSQ equation, 
based on \cite{IVPs}, and we provide an example of a (2,1)-periodic reduction. We show that the growth of the mapping is quadratic and prove that it possesses the Laurent property, and therefore gives rise to (Somos-like) integer sequences. We also show how the conservation laws for the lattice trilinear BSQ give rise to integrals for the reduction. In section 5 we consider a higher lattice Gel'fand-Dikii system. We show how to construct straight-band initial value problems in any direction, and provide explicitly the (2,1)-periodic reduction which is a system of 3 recurrences (of degrees 4,3,3) and equivalent to a 10-dimensional map. We prove that this map also has the Laurent property and therefore is able to generate coupled integer sequences.

\section{Trilinear BSQ as reduction from lattice AKP}

We will show how to derive the trilinear BSQ equation by dimensional reduction from the bilinear AKP equation,
\begin{equation}\label{eq:bilKP} 
(p-q) \wh{\wt{\tau}}\,\ol{\tau} +(q-r) \wh{\ol{\tau}}\,\wt{\tau}+ 
(r-p) \wt{\ol{\tau}}\, \wh{\tau}=0\  , 
\end{equation}
where $p$, $q$, $r$ are three lattice parameters associated with three different lattice shift operators  
$T_p$, $T_q$, $T_r$ respectively in a three-dimensional lattice on which the function $\tau=\tau(n,m,h)$ is defined, i.e. 
$\wt{\tau}=T_p\tau=\tau(n+1,m,h)$, $\wh{\tau}=T_q\tau=\tau(n,m+1,h)$ and $\ol{\tau}=T_r\tau=\tau(n,m,h+1)$. Equation \eqref{eq:bilKP} actually represents an infinite 
parameter-family of compatible equations, each of which lives on a 3-dimensional octahedral sublattice of the infinite-dimensional lattice comprising independent lattice shifts for
each value of $p,q,r$, cf. e.g. \cite{ABS,GRS}. The 3D consistency is also an important instrument in the reduction procedure to 2-dimensional lattice equations. 

We will be looking at periodic reductions under multiple shifts of the $\tau$-function of \eqref{eq:bilKP}. The first such reduction (the case $N=2$ in the framework of \cite{NPCQ}) is obtained by imposing 
\begin{equation}\label{eq:KdVred} 
T_{-p}\circ T_p\tau =\tau 
\end{equation}
for each lattice direction associated with any chosen parameter $p$. In other words, 
$T_{-p}=T_p^{-1}$ represents the opposite, or reverse lattice shift. 
It can be show that the condition \eqref{eq:KdVred} leads to a reduction to the 
bilinear form of the lattice KdV equation. 
In fact, setting $r=-p$ and $r=-q$ in \eqref{eq:bilKP} and identifying the $\ol{\tau}$ 
with $T_p^{-1}\tau=\underaccent{\tilde}{\tau}$ and $T_q^{-1}\tau=\underaccent{\hat}{\tau}$ 
respectively, which are the reverse shifts to the shifts $\wt{\tau}$ and $\wh{\tau}$, we get 
the two 6-point bilinear lattice equations 
\begin{subequations} \label{eq:KdVtau} 
\begin{align}
& (p-q) \wh{\wt{\tau}}\,\underaccent{\tilde}{\tau} +(q+p) \underaccent{\tilde}{\wh{\tau}}\,\wt{\tau}= 
2p \tau\, \wh{\tau}\  , \\ 
& (q-p) \wh{\wt{\tau}}\,\underaccent{\hat}{\tau} +(q+p) \underaccent{\hat}{\wt{\tau}}\,\wh{\tau}= 
2q \tau\, \wt{\tau}\  ,  
\end{align} \end{subequations} 
which are compatible, cf. \cite{HJN,NRGO}, and constitute the bilinear Hirota form of the 
lattice KdV equation. 

In contrast to the reduction from lattice AKP to lattice KdV,  the reduction from lattice AKP to the 
lattice BSQ equation  (the case $N=3$ in the framework of \cite{NPCQ}) is much more 
subtle, and is obtained by imposing instead of \eqref{eq:KdVred} the following condition: 
\begin{equation}\label{eq:BSQred} 
T_{\omega^2 p}\circ T_{\omega p}\circ T_p\tau =\tau\ ,  
\end{equation}
(and similarly for the $q$-shifts) where $\omega=\exp(2\pi i/3)$ is the cube root of unity, and which implies that there is a 
three-fold reversion of each lattice shift. We will show that \eqref{eq:BSQred} leads to a 
reduction to the trilinear BSQ lattice equation of \cite{ZZN}. 

The way to do the analysis is by first concentrating on the $p$-shifts and setting in \eqref{eq:bilKP} $r=\omega p$ and $r=\omega^2 p$ 
respectively (obviously, a similar analysis can subsequently be done for the $q$-shift 
leading to complementary conditions on the reduction). Thus, let us identify the shift $T_{\omega p}\tau=: \ol{\tau}$, which implies that 
$T_{\omega^2 p}\tau=\underaccent{\tilde}{\underaccent{\bar}{\tau}}$. 
This leads to the following set of equations 
\begin{subequations}\label{eq:BSQreds}\begin{align}
(p-q) \wh{\wt{\tau}} \ol{\tau} +(q-\omega p) \wh{\ol{\tau}} \wt{\tau} 
+(\omega -1)p \wt{\ol{\tau}} \wh{\tau}&=0 , \label{eq:BSQredsa}\\ 
(p-q) \wh{\wt{\ol{\tau}}} \underaccent{\tilde}{\tau} +(q-\omega^2 p) \underaccent{\tilde}{\wh{\tau}} 
\wt{\ol{\tau}}+(\omega^2 -1)p \tau \wh{\ol{\tau}}&=0. \label{eq:BSQredsb}
\end{align} \end{subequations} 
The aim is now to eliminate the 'alien shift' $\ol{\tau}=:\vf$ from this system of equations. This can be 
achieved as follows. 
First, let us rewrite \eqref{eq:BSQreds} as a linear system for $\vf$: 
\begin{subequations}
\begin{align}
& A\vf+B\wh{\vf}+C\wt{\vf}=0 \ , \label{eq:ABC} \\ 
& D\wh{\wt{\vf}}+E\wh{\vf}+F\wt{\vf}=0\  , \label{eq:DEF}
\end{align}
\end{subequations}
in which the coefficients are given by\footnote{These coefficients are generalised to the elliptic case in  
\cite{NSZ}.} 
\begin{align}\label{eq:ABCDEF}
&A=(p-q)\wh{\wt{\tau}}\ , \quad B=(q-\omega p)\wt{\tau}\ , \quad C=(\omega-1)p\wh{\tau}\ , \nonumber \\ 
&D=(p-q)\utl{\tau}\ , \quad E=(\omega^2-1)p\tau\ , \quad F=(q-\omega^2 p)\wh{\utl{\tau}}\ , 
\end{align}
and shift \eqref{eq:ABC} and \eqref{eq:DEF} in the $\wh{\phantom{a}}$ and 
$\wt{\phantom{a}}$ directions, while back-substituting $\wh{\wt{\vf}}$. This will give the relations 
\begin{subequations}
\begin{align}
& \left( \wh{A}-\frac{\wh{C}E}{D}\right) \wh{\vf}+\wh{B}\wh{\wh{\vf}}= \frac{\wh{C}F}{D}\wt{\vf}\ , \label{eq:hathat} \\ 
& \left( \wt{A}-\frac{\wt{B}F}{D}\right) \wt{\vf}+\wt{C}\wt{\wt{\vf}}= \frac{\wt{B}E}{D}\wh{\vf}\ , \label{eq:tiltil} 
\end{align}
\end{subequations}
respectively. Next, we need to get rid of the double-shifted objects $\wt{\wt{\vf}}$ and 
$\wh{\wh{\vf}}$. This can be done by applying a $\wh{\phantom{a}}$ shift on \eqref{eq:tiltil} 
and use \eqref{eq:DEF} to rewrite $\wh{\wt{\wt{\vf}}}$ , and subsequently back-substituting
$\wh{\wh{\vf}}$ and $\wt{\wt{\vf}}$ which are obtained from \eqref{eq:hathat} and 
\eqref{eq:tiltil} respectively. This leads to the 
following complicated relation
\[
\left( \wh{\wt{A}} - \frac{\wh{\wt{B}}\wh{F}}{\wh{D}}-\frac{\wh{\wt{C}}\wt{E}}{\wt{D}}\right)
\frac{E\wh{\vf}+F\wt{\vf}}{D} 
+ \frac{\wh{\wt{C}}\wt{F}}{\wt{C}\wt{D}} \left[\frac{\wt{B}E}{D}\wh{\vf} 
-\left( \wt{A}-\frac{\wt{B}F}{D}\right)\wt{\vf}\right]
+ \frac{\wh{\wt{B}}\wh{E}}{\wh{B}\wh{D}} \left[\frac{\wh{C}F}{D}\wt{\vf} 
-\left( \wh{A}-\frac{\wh{C}E}{D}\right)\wh{\vf} \right]=0,  
\]
which only involves $\wh{\vf}$ and $\wt{\vf}$. Assuming that these functions are independent, 
we can split the above relation into two, namely 
\[
\left( \wh{\wt{A}} - \frac{\wh{\wt{B}}\wh{F}}{\wh{D}}-\frac{\wh{\wt{C}}\wt{E}}{\wt{D}}\right)\frac{E}{D} 
+ \frac{\wh{\wt{C}}\wt{F}}{\wt{C}\wt{D}}\frac{\wt{B}E}{D}= 
\frac{\wh{\wt{B}}\wh{E}}{\wh{B}\wh{D}} \left( \wh{A}-\frac{\wh{C}E}{D}\right), 
\quad
\left( \wh{\wt{A}} - \frac{\wh{\wt{B}}\wh{F}}{\wh{D}}-\frac{\wh{\wt{C}}\wt{F}}{\wt{D}}\right)\frac{F}{D} 
+ \frac{\wh{\wt{B}}\wh{E}}{\wh{B}\wh{D}}\frac{\wh{C}F}{D}= 
\frac{\wh{\wt{C}}\wt{F}}{\wt{C}\wt{D}} \left( \wt{A}-\frac{\wt{B}F}{D}\right),
\]
which can be simplified to the following two relations: 
\begin{equation} \label{eq:fundrela}
\frac{\wh{\wt{C}}\wt{F}}{\wt{C}\wt{D}}\frac{\wt{A}}{F}= \frac{\wh{\wt{B}}\wh{E}}{\wh{B}\wh{D}} 
\frac{\wh{A}}{E}\  ,
\quad \wh{\wt{A}} - \frac{\wh{\wt{B}}\wh{F}}{\wh{D}}-\frac{\wh{\wt{C}}\wt{E}}{\wt{D}}
+\frac{\wh{\wt{B}}\wh{E}\wh{C}}{\wh{B}\wh{D}} = \frac{\wh{\wt{C}}\wt{F}}{\wt{C}\wt{D}}
\left( \frac{\wt{A}D}{F}-\wt{B}\right). 
\end{equation}
The first of these relations is trivially satisfied upon inserting the explicit coefficients 
\eqref{eq:ABCDEF}, whilst the second relation yields 
\begin{align}
&(p-q)^2 \wh{\wh{\wt{\wt{\tau}}}}-(p^2+pq+q^2) 
\frac{\wh{\wt{\wt{\tau}}}\wh{\wh{\utl{\tau}}}}{\wh{\utl{\tau}}} 
-3p^2\frac{\wh{\wh{\wt{\tau}}}\,\wt{\tau}}{\tau} 
+3p^2 \frac{ \wh{\wt{\wt{\tau}}}\wh{\wh{\tau}} \wh{\tau}}{\wh{\utl{\tau}} \wh{\wt{\tau}}} 
=(p-q)^2 \frac{\wh{\wh{\wt{\tau}}} \wh{\tau} \wh{\wt{\wt{\tau}}}\utl{\tau}}{\wh{\wt{\tau}} \wh{\utl{\tau}}\tau} -(p^2+pq+q^2) \frac{\wt{\wt{\tau}} \wh{\tau} \wh{\wh{\wt{\tau}}}}{\tau \wh{\wt{\tau}}}\  . \label{eq:midrel} 
\end{align} 
Setting by definition
\[ \Gamma:= (p-q)^2 \wh{\wt{\tau}}\tau\uht{\utl{\tau}} + 
(p^2+pq+q^2) \left(\wh{\utl{\tau}} \wt{\tau}\uht{\tau}+ \wt{\uht{\tau}}\wh{\tau}\utl{\tau} \right)
-3p^2 \wh{\tau}\tau\uht{\tau} -3q^2 \wt{\tau}\tau \utl{\tau}\ , \]
we can rewrite \eqref{eq:midrel} simply as 
\[ \Gamma=\frac{\wt{\tau}}{\utl{\utl{\tau}}} \utl{\Gamma}\  , \]
which can be integrated as 
\begin{align}\label{eq:Gamma} 
\Gamma= \gamma \wt{\tau}\tau\utl{\tau}\  , \end{align} 
where $\gamma$ is independent of the $\wt{\phantom{a}}$ shift, i.e., $\wt{\gamma}=\gamma$. 
Recalling that the entire analysis so far was only taking into account the relations 
\eqref{eq:BSQreds} where we chose $r=\omega p$ and $r=\omega^2 p$ in \eqref{eq:bilKP}. 
Obviously, we can redo the entire analysis by choosing $r=\omega q$ and $r=\omega^2 q$, 
in which case, in addition to the above form for $\Gamma$ we get 
\[ \Gamma= \gamma' \wh{\tau}\tau\uht{\tau}\  , \]
where $\gamma'$ is independent of the $\wh{\phantom{a}}$ shift. Generically both forms for 
$\Gamma$ can only hold true if $\gamma=\gamma'=0$, and this 
yields $\Gamma=0$ which is the trilinear BSQ equation. 

\subsection*{Generalization to the extended BSQ case}

Instead of $\omega$ being a cube root of unity, the extended BSQ cases studied in \cite{ZZN} 
generalize the dispersion curve for the BSQ systems, i.e. the cusp cubic $k^3=p^3$, to a 
dispersion relation of the form\footnote{ In fact, \eqref{eq:dispcurve} represents 
an unfolding of the singular dispersion curve given by the cusp cubic. As a consequence the 
corresponding solutions of the BSQ system in \cite{HietZhang} possess a smoother 
behaviour than the 'pure BSQ' case given by the cube root of unity case. } 
\begin{align}\label{eq:dispcurve}
\omega^3+\alpha\omega^2+\beta\omega = p^3 +\alpha p^2+\beta p\  , 
\end{align}
with roots $p$, $\omega_1(p)$ and $\omega_2(p)$, and where $\alpha$ and $\beta$ are some fixed 
(constant) parameters. Setting now $r=\omega_1(p)$ and $r=\omega_2(p)$ 
in \eqref{eq:bilKP} the whole analysis above goes through unaltered up to \eqref{eq:midrel}, but where 
the coefficients $A,\ldots,E$ are changed into 
\begin{align}\label{eq:extABCDEF}
&A=(p-q)\wh{\wt{\tau}}\ , \quad B=(q-\omega_1(p))\wt{\tau}\ , \quad C=(\omega_1(p)-p)\wh{\tau}\ , \nonumber \\ 
&D=(p-q)\utl{\tau}\ , \quad E=(\omega_2(p)-p)\tau\ , \quad F=(q-\omega_2(p))\wh{\utl{\tau}}\ .  
\end{align}
This change, while leaving the first of the relations \eqref{eq:fundrela} still trivially satisfied, leads to some changes in 
the evaluation of the second. In fact, from the dispersion relation \eqref{eq:dispcurve} 
we can deduce that the following relations hold between the roots $\omega_1(p)$ 
and $\omega_2(p)$, namely 
\[ \omega_1(p)+\omega_2(p)+p=-\alpha\ , \quad \omega_1(p)\omega_2(p)=p^2+\alpha p+\beta\ , \]
which leads to the following extended form of the $\Gamma$ object: 
\[ \Gamma:= (p-q)^2 \wh{\wt{\tau}}\tau\uht{\utl{\tau}} + 
(p^2+pq+q^2+\alpha(p+q)+\beta) 
\left(\wh{\utl{\tau}} \wt{\tau}\uht{\tau}+ \wt{\uht{\tau}}\wh{\tau}\utl{\tau} \right)
-(3p^2+2\alpha p+\beta) \wh{\tau}\tau\uht{\tau} -(3q^2+2\alpha q+\beta) \wt{\tau}\tau \utl{\tau}.
\]
The extended trilinear BSQ equation $\Gamma=0$ is the same as \eqref{eq:TB} with the given coefficients  \cite[Equation (71)]{ZZN}. In terms of $a=p-q$, $b=q- \omega_1(p)$, and $c=q- \omega_2(p)$ it reads
\[
a^2 \wh{\wt{\tau}}\tau\uht{\utl{\tau}} + 
bc 
\left(\wh{\utl{\tau}} \wt{\tau}\uht{\tau}+ \wt{\uht{\tau}}\wh{\tau}\utl{\tau} \right)
-(a + b)(a + c)\wh{\tau}\tau\uht{\tau}
+(ab + ac - bc) \wt{\tau}\tau \utl{\tau}=0,
\]
although from this form it is  not immediately evident that the equation is symmetric under the interchange of $p$ and $q$ and of the corresponding lattice shifts (this follows from the nature of the roots $\omega_i(p)$). 
We note that in \cite{NSZ} elliptic 
solutions of the trilinear BSQ (and other lattice BSQ equations), were constructed, which involve a parametrisation 
in terms of elliptic functions as coefficients, including an 'elliptic cube root of unity'. The 
latter forms another deformation of the pure cusp cubic BSQ lattice. 

\subsection*{Higher Gel'fand-Dikii multilinear system}

We can extend the methods of this section to obtain higher Gel'fand-Dikii (GD) multilinear lattice 
equations by dimensional reduction. The first next higher system (the case $N=4$ in the framework 
of \cite{NPCQ}) is obtained by a four-fold reduction constraint 
\begin{align}\label{eq:N4red} 
T_{\omega_3(p)}\circ T_{\omega_2(p)}\circ T_{\omega_1(p)}\circ T_p\tau=\tau\ , 
\end{align} 
for each lattice direction labelled by the lattice parameter $p$, and where $\omega_i$, $i=1,2,3$, 
are the roots of the quartic polynomial dispersion curve 
\begin{align}\label{eq:disp4curve}
\omega^4+\alpha\omega^3+\beta\omega^2+\gamma\omega = p^4 +\alpha p^3+\beta p^2+\gamma p\  , 
\end{align}
with constants $\alpha$, $\beta$ and $\gamma$. Following the same procedure as before, setting 
subsequently $r=\omega_1(p)$, $r=\omega_2(p)$ and $r=\omega_3(p)$ in \eqref{eq:bilKP} with 
the various 'alien' lattice shifted objects, it is now convenient to choose 
$\vf:=T_{\omega_1(p)}^{-1}\tau$, $\psi:=T_{\omega_2(p)}\tau$, in which case we get the 
following coupled system 
\begin{subequations}\label{eq:vfpsi}\begin{align}
& (p-q)\tau \wh{\wt{\vf}}+(q-\omega_1(p))\wh{\tau}\wt{\vf}+(\omega_1(p)-p)\wt{\tau}\wh{\vf}=0 \ , 
\label{eq:vfpsia}\\ 
& (p-q)\wh{\wt{\tau}}\psi + (q-\omega_2(p))\wt{\tau}\wh{\psi} +(\omega_2(p)-p) \wh{\tau}\wt{\psi}=0\ , 
\label{eq:vfpsib}\\ 
& (p-q)\utl{\vf}\wh{\wt{\psi}} +(q-\omega_3(p)) \wh{\utl{\vf}}\wt{\psi} + (\omega_3(p)-p)\vf\wh{\psi}=0\ , 
\label{eq:vfpsic}
\end{align} 
\end{subequations} 
where the first two equations are linear, in $\vf$ and $\psi$ respectively, but where the 
last equation couples the other ones. Because of the nonlinearity of \eqref{eq:vfpsic}, the analysis
is distinctly more complicated in this case, however, we note that the two equations 
\eqref{eq:vfpsib} and \eqref{eq:vfpsic} exhibit the same structure as the system 
comprising \eqref{eq:ABC} and \eqref{eq:DEF} of the trilinear case, but now for the function $\psi$. 
Thus, we can use the exact same procedure as in the trilinear case to eliminate $\psi$, i.e. 
use the equations \eqref{eq:fundrela}, where now the coefficients $A,\ldots,F$ are identified as 
\begin{align}\label{eq:N4ABCDEF}
&A=(p-q)\wh{\wt{\tau}}\ , \quad B=(q-\omega_2(p))\wt{\tau}\ , \quad C=(\omega_2(p)-p)\wh{\tau}\ , \nonumber \\ 
&D=(p-q)\utl{\vf}\ , \quad E=(\omega_3(p)-p)\vf\ , \quad F=(q-\omega_3(p))\wh{\utl{\vf}}\ .  
\end{align}
Once again, the first relation of \eqref{eq:fundrela} is trivially satisfied with these coefficients, while 
the second gives rise to the a bilinear equation in both $\tau$ and $\vf$, namely 
\begin{align} 
& (p-q)^2 \left(\wh{\wh{\wt{\wt{\tau}}}}\,\wh{\wt{\tau}}\,\vf\,\wh{\utl{\vf}} - 
\wh{\wh{\wt{\tau}}}\,\wh{\wt{\wt{\tau}}}\,\utl{\vf}\wh{\vf}\right) 
 +(p-\omega_3(p))(p-\omega_2(p)) 
\left(\wh{\wt{\wt{\tau}}}\, \wh{\wh{\tau}}\,\wh{\vf}\,\vf - 
\wh{\wh{\wt{\tau}}}\,\wh{\wt{\tau}}\,\wt{\vf}\,\wh{\utl{\vf}} \right)\ \nonumber \\ 
&+ (q - \omega_3(p))(q - \omega_2(p))
\left( \wt{\wt{\tau}}\,\wh{\wh{\wt{\tau}}}\,\wh{\vf}\,\wh{\utl{\vf}} 
- \wh{\wt{\wt{\tau}}}\,\wh{\wt{\tau}}\,\vf\,\wh{\wh{\utl{\vf}}} \right)=0\ .    \label{eq:vfvf} 
\end{align} 
Thus we get a coupled system for $\vf$ comprising the linear equation \eqref{eq:vfpsia} 
together with the bilinear equation \eqref{eq:vfvf}, from which we would still like to 
eliminate $\vf$. However, we will refrain from that particular computation for the present paper\footnote{It is worth mentioning that there is a companion system with another function 
$\chi=T_{\omega_1(q)}^{-1}\tau$ where the root $\omega_1(q)$ appears  (from the 
dispersion curve \eqref{eq:disp4curve} with $p$ replaced by $q$), and where the roles of 
$p$ and $q$ and the corresponding shifts $\wt{\phantom{a}}$ and $\wh{\phantom{a}}$ are 
reversed.}.  We just mention here, without proof, that from the direct linearisation method 
designed in \cite{ZZN}, and elaborated in \cite{TZZ} for the case $N=4$, the following 10-term 
16-point sextic equation for the $\tau$-, $\vf$- and $\psi$-functions can be derived: 
\begin{align}
& (p-q)^3\, \wt{\tau}\wh{\tau}\left(\tau\uht{\utl{\tau}} \wh{\wh{\wt{\tau}}}\wh{\wt{\wt{\tau}}}  
- \uht{\tau}\utl{\tau}\wh{\wt{\tau}}\wh{\wh{\wt{\wt{\tau}}}} \right)= (4p^3+3\alpha p^2+2\beta p+\gamma) \wh{\tau}\wt{\tau}
\left(\uht{\tau}\tau \wh{\wh{\tau}}\wh{\wt{\wt{\tau}}}-\utl{\tau}\uht{\wt{\tau}}\wh{\wt{\tau}} 
\wh{\wh{\wt{\tau}}}\right)  \nonumber \\ 
& -(4q^3+3\alpha q^2+2\beta q+\gamma) \wh{\tau}\wt{\tau}
\left(\utl{\tau}\tau \wt{\wt{\tau}}\wh{\wh{\wt{\tau}}}-\uht{\tau}\utl{\wh{\tau}}\wh{\wt{\tau}} 
\wh{\wt{\wt{\tau}}}\right)
+ \left[ (p^3+p^2q+pq^2+q^3) +\alpha(p^2+pq+q^2)+\beta(p+q)+\gamma \right] \nonumber \\  
& \quad \times\left[ \utl{\tau}\wh{\tau}\wh{\wh{\wt{\tau}}}
\left(\tau\wt{\wt{\uht{\tau}}}\wh{\wt{\tau}}+\wt{\uht{\tau}}\wh{\tau}\wt{\wt{\tau}} \right) 
- \uht{\tau}\wt{\tau}\wh{\wt{\wt{\tau}}}
\left(\tau\wh{\wh{\utl{\tau}}}\wh{\wt{\tau}}+\wh{\utl{\tau}}\wt{\tau}\wh{\wh{\tau}} \right) 
\right], \label{eq:hexaGD}
\end{align}
which was absent from \cite{TZZ}. Due to its complicated structure we will for now abstain from 
doing further analysis on this equation (including the question of whether this equation can be 
seen as a consequence of a pair of 7-term 12-point quadrilinear equations that one would 
expect to govern the $N=4$ $\tau$-function structure).

\section{A dual to the trilinear Boussinesq equation, and a matrix conservation law}

Taking $D=0$ in the trilinear Boussinesq equation \eqref{eq:TB} and dividing by $\tau$ yields
\begin{equation} \label{eq:RAKP}
A\wh{\tau}\uht{\tau}+B\wt{\tau} \utl{\tau}+C\wh{\wt{\tau}}\uht{\utl{\tau}}=0.
\end{equation}
This 2D lattice equation can be obtained from the 3D AKP equation \eqref{eq:bilKP} by the following reduction:
\begin{equation} \label{3t2}
\wt{\tau}\mapsto \uht{\utl{\tau}},\quad
\wh{\tau}\mapsto \wh{\tau},\quad
\ol{\tau}\mapsto \wt{\tau},
\end{equation}
whilst setting $A=r-p$, $B=p-q$, $C=q-r$. Thus, if $\tau(k,l,m)$ satisfies \eqref{eq:bilKP} then $\tau(m-k,l-k)$ satisfies \eqref{eq:RAKP}.

In \cite{D2}, the following dual to the AKP equation \eqref{eq:bilKP} was derived,
\begin{equation} \label{DAKP}
a_{{1}} \left(\utl{\wh{\ol{\tau}}}\wt{\tau}\ol{\wt{\tau}}\wh{\wt{\tau}}
-\ol{\tau}\wh{\tau}\wh{\ol{\tau}}\wt{\wt{\tau}} \right)
+ a_{{2}} \left(\wh{\tau}\wh{\ol{\tau}}\ol{\wt{\uht{\tau}}}\wh{\wt{\tau}}
-\ol{\tau}\wh{\wh{\tau}}\wt{\tau}\ol{\wt{\tau}} \right)
+ a_{{3}} \left(
\ol{\tau}
\wh{\ol{\tau}}
\ol{\wt{\tau}}
\wh{\wt{\udo{\tau}}}
-\ol{\ol{\tau}}
\wh{\tau}
\wt{\tau}
\wh{\wt{\tau}}
\right)
+ a_{{4}} \left( \tau\wh{\ol{\tau}}\ol{\wt{\tau}}\wh{\wt{\tau}}
-\ol{\tau}\wh{\tau}\wt{\tau}\ol{\wh{\wt{\tau}}} \right)=0,
\end{equation}
employing the characteristics (denoted $\Lambda_1,\Lambda_2,\Lambda_3,\Lambda_7$ in \cite{MQ}),
\begin{equation} \label{W}
\W=\left(\dfrac{\hat{\ol{\utl{\tau}}}}{\hat{\ol{\tau}}\hat{\tau}\ol{\tau}}-
\dfrac{\tilde{\tilde{\tau}}}{\tilde{\tau}\tilde{\hat{\tau}}\tilde{\ol{\tau}}},\
\dfrac{\tilde{\ol{\uht{\tau}}}}{\tilde{\ol{\tau}}\tilde{\tau}\ol{\tau}}-
\dfrac{\hat{\hat{\tau}}}{\hat{\tau}\hat{\ol{\tau}}\tilde{\hat{\tau}}},\
\dfrac{\udo{\tilde{\hat{\tau}}}}{\tilde{\hat{\tau}}\tilde{\tau}\hat{\tau}}-
\dfrac{\ol{\ol{\tau}}}{\ol{\tau}\tilde{\ol{\tau}}\hat{\ol{\tau}}},\
\dfrac{\tau}{\tilde{\tau}\hat{\tau}\ol{\tau}}-
\dfrac{\tilde{\hat{\ol{\tau}}}}{\tilde{\hat{\tau}}\hat{\ol{\tau}}\tilde{\ol{\tau}}}
\right).
\end{equation}
Reductions of the dual equation \eqref{DAKP} include Rutishauser's quotient-difference (QD) algorithm, the higher analogue of the discrete time Toda (HADT) equation and its corresponding quotient-quotient-difference (QQD) system, the discrete hungry Lotka-Volterra system, discrete hungry QD, as well as the hungry forms of HADT and QQD. In \cite{D2}, it was conjectured that \eqref{DAKP} has the Laurent property, vanishing algebraic entropy and that it admits N-soliton solutions. The latter was established in \cite{KZQ} by relating it to a 14-point equation found by King and Schief \cite{KS}, which itself is a consequence of the lattice BKP equation (also known as the Miwa equation).
From three of the four characteristics \eqref{W}, we obtain three characteristics for \eqref{eq:TB}. As we shall see, $W_1$ does not provide us with a characteristic for \eqref{eq:TB}. Applying the reduction \eqref{3t2} to \eqref{W}, and dividing by $-\tau$ gives us
\[
X=\left(
\frac{\uht{\uht{\utl{\utl{\tau}}}}}{\tau\utl{\tau}\uht{\tau}\uht{\utl{\tau}}}-\frac{\wh{\wh{\wt{\wt{\tau}}}}}{\tau\wt{\tau}\wh{\tau}\wh{\wt{\tau}}},
\frac{\wh{\wh{\tau}}}{\wh{\tau} \tau \utl{\tau} \wh{\wt{\tau}}}-\frac{\uht{\uht{\tau}}}{\tau \uht{\tau} \wt{\tau} \utl{\uht{\tau}}},
 \frac{\wt{\wt{\tau}}}{\tau \uht{\tau} \wt{\tau} \wh{\wt{\tau}}}-\frac{\utl{\utl{\tau}}}{\wh{\tau} \tau \utl{\uht{\tau}} \utl{\tau}},
\frac{1}{\uht{\tau} \utl{\tau} \wh{\wt{\tau}}}-\frac{1}{\wh{\tau} \wt{\tau} \utl{\uht{\tau}}} 
\right),
\]
which are four characteristics for equation \eqref{eq:TB} when $D=0$, i.e. $
\tau(A\wh{\tau}\uht{\tau}+B\wt{\tau} \utl{\tau}+C\wh{\wt{\tau}}\uht{\utl{\tau}})=0$.
The question is whether they are also characteristics for the remaining term of \eqref{eq:TB}, $T=D\left(\wh{\utl{\tau}} \wt{\tau} \uht{\tau}+ \wt{\uht{\tau}} \wh{\tau} \utl{\tau}\right)$. One can verify whether an expression $Z[\tau]$ is a divergence, by checking whether it is in kernel of the discrete Euler operator, i.e. whether
\[
E(Z)=\sum_\sigma \sigma^{-1}\left(\frac{\partial Z[\tau]}{\partial \sigma(\tau)}\right)=0,
\]
where the sum is over all applicable shifts on the lattice, cf. \cite{HM}. Surprisingly, $E(TX_i)=0$ for $i=2,3,4$. Thus, by taking a linear combination, and multiplying by the common denominator we obtain a quadrilinear dual equation,
\begin{equation} \label{DTB}
x(\wh{\tau} \utl{\tau} \utl{\uht{\tau}} \wt{\wt{\tau}}-\uht{\tau} \wt{\tau} \wh{\wt{\tau}} \utl{\utl{\tau}})+y( \uht{\tau} \wt{\tau} \utl{\uht{\tau}} \wh{\wh{\tau}}-\wh{\wt{\tau}} \wh{\tau} \utl{\tau} \uht{\uht{\tau}}) +z( \tau \wt{\tau} \wh{\tau} \utl{\uht{\tau}}-\tau \uht{\tau} \wh{\wt{\tau}} \utl{\tau})=0.
\end{equation}

Comparing the stencil on which equation \eqref{DTB} is defined, depicted in  Figure \ref{stencil}, with the stencil of the HADT equation, cf. \cite[Figure 1]{SNV}, it is clear which transformation on the independent variables could relate the two equations. Indeed, if $\sigma(l,k)$ satisfies the HADT equation
\[
\uht{\uht{\wt{\sigma}}} \utl{\hat{\sigma}} \left(\uht{\sigma} \hat{\hat{\sigma}}-\sigma \hat{\sigma}\right)
 = 
\uht{\sigma} \utl{\hat{\hat{\sigma}}} \left(\utl{\hat{\sigma}} \uht{\uht{\wt{\wt{\sigma}}}}-\sigma \uht{\wt{\sigma}}\right)+\hat{\sigma} \uht{\wt{\sigma}} \left(-\utl{\utl{\hat{\hat{\sigma}}}} \uht{\uht{\wt{\sigma}}}+\utl{\hat{\hat{\sigma}}} \uht{\uht{\sigma}}\right),
\]
then $\tau(l,k)=\sigma(l,-l-k)$ satisfies equation \eqref{DTB}, with $x=y=1$, $z=-1$.
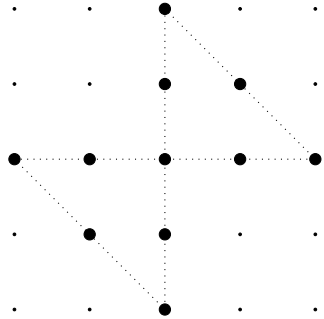
\begin{figure}[h]
\begin{center}
\begin{tikzpicture}[scale=1]
\foreach \x in {1,...,5}{
      \foreach \y in {1,...,5}{
        \node[draw,circle,inner sep=0pt,fill] at (\x,\y) {};
      }
    }
\foreach \x in {(1,3),(2,2),(2,3),(3,1),(3,2),(3,3),(3,4),(3,5),(4,3),(4,4),(5,3)}{
    \node[draw,circle,inner sep=1.5pt,fill] at \x {};
    }
\draw[dotted] (1,3)--(5,3)--(3,5)--(3,1)--(1,3);
\end{tikzpicture}
\end{center}
\caption{\label{stencil} The stencil of the dual  equation (\ref{DTB}).}
\end{figure}

A matrix conservation law is given by
\begin{equation} \label{CONS}
\tilde{\P}-\P+\hat{\Q}-\Q=\V^T U,
\end{equation}
where $\P,\Q$ are the $4\times3$ matrices
\[
\P=
\begin{pmatrix}
\dfrac{\wh{\utl{\tau}} \wt{\tau}}{\tau \wh{\tau}}+\dfrac{\utl{\utl{\wh{\tau}}} \tau}{\utl{\tau} \wh{\utl{\tau}}} & 0 & 0
\\[2mm]
 \dfrac{\wt{\tau} \utl{\utl{\tau}}}{\wh{\tau} \utl{\uht{\tau}}} & \dfrac{\tau \wh{\wh{\utl{\tau}}}}{\wh{\tau} \wh{\utl{\tau}}} & \dfrac{\utl{\tau} \tau}{\wh{\tau} \utl{\uht{\tau}}}
\\[2mm]
 \dfrac{\utl{\utl{\uht{\tau}}} \wt{\tau}}{\utl{\uht{\tau}} \tau} & -\dfrac{\utl{\uht{\tau}} \wh{\wh{\tau}}}{\wh{\tau} \utl{\tau}} & -\dfrac{\utl{\uht{\tau}} \tau}{\uht{\tau} \utl{\tau}}
\\[2mm]
 \dfrac{\wt{\tau} \utl{\utl{\wh{\tau}}}}{\utl{\tau} \wh{\tau}}+\dfrac{\uht{\tau} \utl{\utl{\tau}} \wt{\tau} \wh{\utl{\tau}}}{\wh{\tau} \tau \utl{\uht{\tau}} \utl{\tau}} & \dfrac{\uht{\tau} \wh{\wh{\utl{\tau}}}}{\wh{\tau} \utl{\tau}} & \dfrac{\uht{\tau} \wh{\utl{\tau}}}{\wh{\tau} \utl{\uht{\tau}}}
\end{pmatrix},
\quad
\Q=\begin{pmatrix}
\dfrac{\utl{\utl{\tau}} \uht{\tau}}{\utl{\uht{\tau}} \utl{\tau}} & \dfrac{\wh{\tau} \uht{\uht{\tau}}}{\wt{\tau} \utl{\uht{\tau}}} & \dfrac{\tau \uht{\tau}}{\wt{\tau} \utl{\uht{\tau}}}
\\[2mm]
 0 & \dfrac{\uht{\tau} \wh{\utl{\tau}}}{\tau \utl{\tau}}+\dfrac{\utl{\tau} \uht{\uht{\tau}}}{\uht{\tau} \utl{\uht{\tau}}} & 0
\\[2mm]
 -\dfrac{\utl{\utl{\uht{\tau}}} \wt{\tau}}{\utl{\uht{\tau}} \tau} & \dfrac{\wh{\wt{\tau}} \uht{\uht{\tau}}}{\uht{\tau} \wt{\tau}} & -\dfrac{\uht{\tau} \wt{\tau}}{\wt{\uht{\tau}} \tau}
\\[2mm]
 \dfrac{\utl{\utl{\tau}} \wt{\uht{\tau}}}{\tau \utl{\uht{\tau}}} & \dfrac{\uht{\uht{\tau}} \wh{\utl{\tau}}}{\tau \utl{\uht{\tau}}}+\dfrac{\wh{\tau} \uht{\uht{\tau}} \utl{\tau} \wt{\uht{\tau}}}{\tau \uht{\tau} \wt{\tau} \utl{\uht{\tau}}} & \dfrac{\utl{\tau} \wt{\uht{\tau}}}{\wt{\tau} \utl{\uht{\tau}}}
\end{pmatrix},
\]
$U=(X_2,X_3,X_4)$ and $\V^T$ denotes the transpose of
$
\V=\left(
\wh{\tau} \tau \uht{\tau},
\wt{\tau} \tau \utl{\tau},
\wh{\wt{\tau}} \tau \utl{\uht{\tau}},
\wt{\tau} \uht{\tau} \wh{\utl{\tau}}+\wh{\tau} \utl{\tau} \wt{\uht{\tau}}
\right).
$
Denoting two vectors of coefficients by $S=\left( A , B , C, D \right)$ and $\w=\left(y,x,z \right)$, we have that
$S\V^T=0$ represents the trilinear Boussinesq equation \eqref{eq:TB} and the equation $U\w^T=0$ is equivalent to dual \eqref{DTB}. Hence, pre-multiplying \eqref{CONS} with $S$  gives three conservation laws for the trilinear Boussinesq  equation, and post-multiplying \eqref{CONS} with $\w^T$ yields four conservation laws for (a rational version of) equation \eqref{DTB}.

In \cite{D2} it was conjectured that reductions of (\ref{DAKP}), such as equation (\ref{DTB}), possess the Laurent property, and have vanishing algebraic entropy. It seems that this also holds true for the trilinear Boussinesq equation itself! 

\section{Periodic reduction, growth and Laurentness of trilinear BSQ, integrals from conservation laws}
A periodic reduction of a lattice equation is a mapping, whose iterates provide a periodic solution. Geometrically, one rolls up the lattice to form a cylinder, cf. \cite[Figure 1]{Orm}. Algebraically, one chooses a vector $\b{s}\in \mathbb{Z}^2$, compatible with the lattice equation, and imposes the periodicity condition $\tau(l,k)=\tau((l,k)+\b{s})$.

How to pose initial values problems (IVPs, Cauchy problems) for lattice equations, on arbitrary stencils, was described in \cite{IVPs}. We summarise the construction and apply it to the trilinear BSQ.

For a given stencil $S$, one defines the $S$-directions as the directions of the edges of the convex hull of $S$. For each $\b{s}$ whose direction is not an $S$-direction\footnote{If the direction $\b{\hat{s}}$ is an $S$-direction, then one has to augment the IVP with values in an additional direction. Such a problem is called a Goursat problem.} the $\b{s}$-periodic reduction is a mapping of dimension
\[
D(\b{s},\b{d}):=\left|\det\left(\begin{array}{c} \b{s} \\ \b{d} \end{array}\right)\right|,
\]
where $\b{d}$ is the difference between two points $p_1\in l_1,p_2\in l_2$ on two lines $l_1,l_2$ with direction $\b{\hat{s}}$ which squeeze the stencil. The mapping corresponds to a shift on the lattice $\b{a}\rightarrow\b{c}$, where $\b{c}$ is $(0,1)$ or $(1,0)$ or the unique lattice point inside the parallelogram spanned by  $\b{\hat{s}}$ and $(1,0)$ such that $D(\b{s},\b{c})=1$.

For a low-dimensional example we give a formula for its growth and prove that it has the Laurent property. 

\subsection*{(2,1)-reduction of trilinear BSQ}
We squeeze the 9-point square stencil using two lines with direction $\b{s}=(2,1)$, see Figure \ref{S1}. The difference between the points where the lines touch the stencil is $\b{d}=(2,-2)$, or alternatively one could e.g. choose points on the lines such that the difference is $\b{d}=(0,3)$. The dimension of the reduction is \[
D(\b{s},\b{d})=\left|\det\begin{pmatrix} 2 & 1 \\ 2 & -2\end{pmatrix}\right|=6.
\]
We take $\b{c}=(1,0)$ (so that $D(\b{s},\b{c})=1$). Now we start labeling the initial values in steps of $\b{c}$, whilst making horse-jumps (-$\b{s}$) at the right boundary of the stencil, see Figure \ref{S2}.

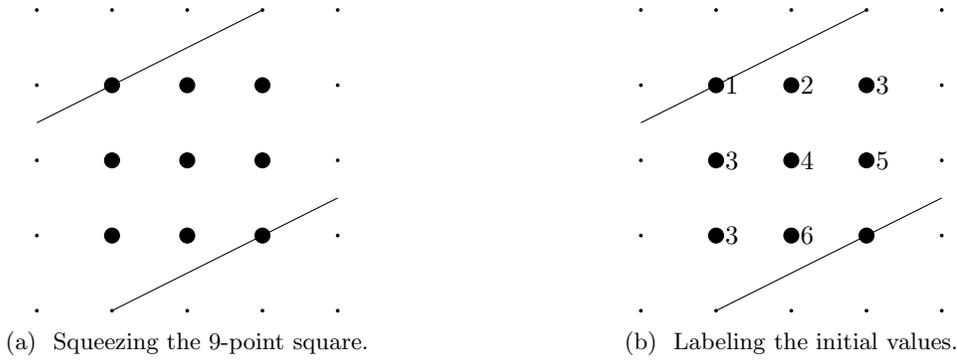
\begin{figure}[h!]
\begin{subfigure}{0.45\textwidth}
\begin{center}
\begin{tikzpicture}[scale=1]
\foreach \x in {1,...,5}{
      \foreach \y in {1,...,5}{
        \node[draw,circle,inner sep=0pt,fill] at (\x,\y) {};
      }
    }
\foreach \x in {2,...,4}{
      \foreach \y in {2,...,4}{
        \node[draw,circle,inner sep=2pt,fill] at (\x,\y) {};
      }
    }
\draw (1,3.5)--(4,5);
\draw (2,1)--(5,2.5);
\end{tikzpicture}
\caption{\label{S1} Squeezing the 9-point square.}
\end{center}
\end{subfigure}
\begin{subfigure}{.45\textwidth}
\begin{center}
\begin{tikzpicture}[scale=1]
\foreach \x in {1,...,5}{
      \foreach \y in {1,...,5}{
        \node[draw,circle,inner sep=0pt,fill] at (\x,\y) {};
      }
    }
\foreach \x in {2,...,4}{
      \foreach \y in {2,...,4}{
        \node[draw,circle,inner sep=2pt,fill] at (\x,\y) {};
      }
    }
\draw (1,3.5)--(4,5);
\draw (2,1)--(5,2.5);
\draw (2,4) node[right] {1} 
    (3,4) node[right] {2}
    (4,4) node[right] {3}
    (2,3) node[right] {3}
    (3,3) node[right] {4}
    (4,3) node[right] {5}
    (2,2) node[right] {3}
    (3,2) node[right] {6};
(3,2) node[right] {6};
\end{tikzpicture}
\caption{\label{S2} Labeling the initial values.}
\end{center}
\end{subfigure}
\caption{\label{rtB} The (2,1)-reduction of trilinear Boussinesq gives rise to a six-dimensional mapping.}
\end{figure}

The mapping takes the form
\begin{subequations}  \label{mg}\begin{equation}\label{eq:mga} 
(x_1,x_2,x_3,x_4,x_5,x_6)\mapsto (x_2,x_3,x_4,x_5,x_6,x_7),    
\end{equation}
where $x_7$ is determined by taking $n=1$ in the recursion 
\begin{equation}\label{eq:mgb}
\big( Ax_{n+1}x_{n+5}+(B+C)x_{n+2}x_{n+4}\big)x_{n+3} + D\big(x_nx_{n+4}x_{n+5}+x_{n+1}x_{n+2}x_{n+6} \big)
=0.
\end{equation} \end{subequations} 
As $C,D$ can be absorbed into $A, B$ we take, without loss of generality, $C=0,D=1$. Starting with initial values linear in a variable $z$, e.g. such as $p=[3 - 2z, 15 - 4z, -5z + 1, 5 + 3z, 1 - z, 8 + 3z]$, the $z$-degree of the numerator of the last component of the $n$-th iterate is given by
\[
\frac{1}{4}\left(n^2+3n+1+(-1)^{\lfloor(n+2)/2\rfloor}\right),
\]
at least for the first 16 iterates we computed.

If for all $n$, the denominator of the $n$-th iterate of a map is a product of powers of the initial values, then the map is said to have the Laurent property. For the above map the denominators are monomials in the initial values (but not the first nor the last one). And this can be proven as follows.

\begin{proposition} \label{lp1}
The mapping (\ref{mg}) has the Laurent property.
\end{proposition}
\begin{proof}
For general initial values, $p_0,\ldots,p_5$ we write the first 7 iterates as $p_i/q_i$, with $(p_i,q_i)=1$ (co-prime), $i=6,7,\ldots,12$. The first few numerators are
\begin{align*}
p_6&=- p_{0} p_{4} p_{5}-A p_{1} p_{3} p_{5}-B p_{2} p_{3} p_{4}\\
p_7&=p_{0} p_{1} p_{4} p_{5}^{2}+A  p_{0} p_{2} p_{4}^{2} p_{5}+A  p_{1}^{2} p_{3} p_{5}^{2}+A^{2} p_{1} p_{2} p_{3} p_{4} p_{5}+A B p_{2}^{2} p_{3} p_{4}^{2}\\
p_8&= p_{0}^{2} p_{1} p_{4} p_{5}^{3}+A p_{0}^{2} p_{2} p_{4}^{2} p_{5}^{2}+A p_{0} p_{1}^{2} p_{3} p_{5}^{3}+ \left(A^{2}+2  B +2  C \right) p_{0} p_{1} p_{2} p_{3} p_{4} p_{5}^{2}+A B^{2} p_{2}^{3} p_{3}^{2} p_{4}^{2}\\
&\ \ \ +2 A  B p_{0} p_{2}^{2} p_{3} p_{4}^{2} p_{5}+2 A  B p_{1}^{2} p_{2} p_{3}^{2} p_{5}^{2}+B \left(A^{2}+ B + C \right) p_{1} p_{2}^{2} p_{3}^{2} p_{4} p_{5}\\
p_9&=- p_{0}^{3} p_{1} p_{4}^{2} p_{5}^{4}-A p_{0}^{3} p_{2} p_{4}^{3} p_{5}^{3}-2 A p_{0}^{2} p_{1}^{2} p_{3} p_{4} p_{5}^{4}- \left(2 A^{2}+ B + C \right) p_{0}^{2} p_{1} p_{2} p_{3} p_{4}^{2} p_{5}^{3}-A^{2} p_{0} p_{1}^{3} p_{3}^{2} p_{5}^{4}\\
&\ \ \ -A B p_{0}^{2} p_{2}^{2} p_{3} p_{4}^{3} p_{5}^{2}-A  \left(A^{2}+2 B +2  C \right) p_{0} p_{1}^{2} p_{2} p_{3}^{2} p_{4} p_{5}^{3}-A^{2}  B p_{0} p_{1} p_{2}^{2} p_{3}^{2} p_{4}^{2} p_{5}^{2}+A B^{3} p_{2}^{4} p_{3}^{3} p_{4}^{3}\\
&\ \ \ +A  B^{2} p_{0} p_{2}^{3} p_{3}^{2} p_{4}^{3} p_{5}-A^{2}  B p_{1}^{3} p_{2} p_{3}^{3} p_{5}^{3}+A^{2} B^{2} p_{1} p_{2}^{3} p_{3}^{3} p_{4}^{2} p_{5}\\
&\ \ \vdots 
\end{align*}
and the denominators are
\[
q_6=p_{1} p_{2},\ q_7=p_{1} p_{2}^{2} p_{3},\ q_8= p_{1}^{2}p_{2}^{2}  p_{3}^{2},\ q_9= 
p_{1}^{2}p_{2}^{4}  p_{3}^{2} p_{4},\ q_{10}= 
p_{1}^{3}p_{2}^{5}  p_{3}^{4} p_{4},\ q_{11}= 
p_{1}^{4} p_{2}^{6} p_{3}^{5} p_{4}^{2},\ q_{12}= 
p_{1}^{5} p_{2}^{8} p_{3}^{6}p_{4}^{2}.
\]
We have verified that the $q_i$ are monomials in $p_1,p_2,p_3,p_4$ and that $(p_6,p_i)=1$ for $i=7,\ldots,12$. This establishes the Laurent property, cf. \cite[Theorem 2]{HHVQ}, that all iterates are Laurent polynomials in $p_1,p_2,p_3,p_4$. 
\end{proof}

Hence, if one starts with initial values $p_1=p_2=p_3=p_4=1$, one obtains a polynomial (or integer) sequence. E.g. for $A=-5,B=-3$ and $p_0=4,p_5=2$, we obtain
\begin{align}\label{eq:Somseq} 
& \ldots,1,2,5, 21, 135, 585, 8640, 228825, 2193075, 72444375, 7923227625, 265006991250, 15144850614375,\ldots, \nonumber \\ 
& 
\end{align}
which is an example of an integer sequence obtained from the trilinear recursion equation 
\eqref{eq:mgb}.
As pointed out by M. Somos \cite{Somos}, it also satisfies the bilinear recursion relation 
\begin{equation}
p_{n+4}p_{n-3} = -3p_{n+3}p_{n-2} + 99p_{n+1}p_n. \label{eq:Somrecurs}
\end{equation}
A similar statement holds true for all values of $A,B$, $p_0=a,p_1=p_2=p_3=p_4=1,p_5=b$. With extra initial value $p_6=-b( A + a) - B$, the (polynomial) iterates of the map \eqref{mg} are generated by
\[
p_{n+4}p_{n-3} = B p_{n+3}p_{n-2} + Y p_{n+1}p_n,
\]
where $Y = ab(A^2 + ab + B) + (a + b)(ab + B)A + B^2$. It remains an open question whether 
integer sequences arising from higher periodic reductions of the trilinear lattice equation \eqref{eq:TB} are all of Somos type (i.e., admit a bilinear recurrence). 

\subsection*{Integrals from conservation laws}
The initial values in Figure \ref{S2} are labeled according to the values of a new variable, in terms of the lattice variables $l,m$ the new variable is $n=l-2m$. More generally, if $\b{s}$ is in the first quadrant, and gcd$(s_1,s_2)=1$, then the variable $n(l,m)=s_2l -s_1m$ has the properties that $n((l,m)+k\b{s})=n(l,m)$ and, with $D(\b{s},\b{c})=1$, $n((l,m)+\b{c})=n(l,m)+1$.

A conservation law for a lattice equation
\[
\wt{P}-P+\wh{Q}-Q\equiv 0\qquad \text{(modulo the equation)}
\]
reduces to
\begin{equation} \label{rcl}
P_{n+s_2}-P_n+Q_{n-s_1}-Q_n\equiv 0\qquad  \text{(modulo the reduction)},    
\end{equation}
and this gives rise to an integral. When $\b{s}$ is in the fourth quadrant, $s_1>0,s_2<0$, and gcd$(s_1,s_2)=1$, we take $n(l,m)=|s_2|l + s_1m$, cf. \cite{IVPs}, and then the conservation law reduces to
\begin{equation} \label{rcl2}
P_{n+|s_2|}-P_n+Q_{n+s_1}-Q_n\equiv 0\qquad  \text{(modulo the reduction)}.    
\end{equation}

\begin{lemma} \label{lem}
The expression
$
K_n=\sum_{i=0}^{s_2-1} P_{n+i} - \sum_{j=1}^{s_1} Q_{n-j}
$
does not depend on $n$ if (\ref{rcl}) holds. The expression
$
L_n=\sum_{i=0}^{|s_2|-1} P_{n+i} + \sum_{j=0}^{s_1-1} Q_{n+j}
$
does not depend on $n$ if (\ref{rcl2}) holds.
\end{lemma}

\begin{proof}
We have $\begin{aligned}[t]
K_{n+1}&=\sum_{i=0}^{s_2-1} P_{n+i+1} - \sum_{j=1}^{s_1} Q_{n-j+1}=\sum_{i=0}^{s_2-2} P_{n+i+1} - \sum_{j=2}^{s_1} Q_{n-(j-1)} + P_{n+s_2} - Q_{n} \\
&=\sum_{i=1}^{s_2-1} P_{n+i} - \sum_{j=1}^{s_1-1} Q_{n-j} + P_{n} - Q_{n-s_1}=\sum_{i=0}^{s_2-1} P_{n+i} - \sum_{j=1}^{s_1} Q_{n-j}=K_n.
\end{aligned}$

The proof for the second statement is similar.
\end{proof}

Considering the three conservation laws we obtain by pre-multiplying \eqref{CONS} with $\left( A , B , 0, 1 \right)$, and Lemma \ref{lem} to construct integrals for the mapping (\ref{mg}, the first conservation law yields a constant (a function of $B$ only), the second gives an expression which depends on $x_7$ and $x_8$ so we rewrite the invariant function in terms of the initial values (and we have subtracted a constant), to get
\begin{equation} \label{fi}
 \left(A^{2}+B\right)\frac{x_{1} x_{6}}{x_{3} x_{4}}+AB\left(\frac{ x_{1} x_{5} }{x_{2} x_{4}}+\frac{ x_{2} x_{6}}{x_{3} x_{5}}\right)+A \frac{x_{1} x_{6}}{x_{3} x_{4}}\left(\frac{x_{5} x_{1}}{x_{2} x_{4}}+\frac{x_{2} x_{6}}{x_{3} x_{5}} \right)
+\left( \frac{x_{1} x_{6}}{x_{3} x_{4}}\right)^{2}.
\end{equation}
The third conservation law directly yields
\begin{equation} \label{si}
\frac{B x_{3} x_{4}+ x_{6} x_{1}}{x_{2} x_{5}}-\frac{A x_{2} x_{4}+ x_{1} x_{5}}{x_{3}^{2}}-\frac{A x_{3} x_{5}+ x_{2} x_{6}}{x_{4}^{2}}.
\end{equation}
In terms of variables $x_i=y_iy_{i+1}^{-2}y_{i+2}$ the integral (\ref{fi}) is polynomial
\[
\left(A^{2}+B\right) y_{1} y_{2}^{2} y_{3}^{2} y_{4} +AB y_{2} y_{3} \left( y_{1} + y_{4} \right)+A y_{1}y_{2}^{3}y_{3}^{3}y_{4}\left( y_{1}   +y_{4} \right) + \left(y_{1}y_{2}^{2} y_{3}^{2} y_{4}\right)^{2}
\]
and (\ref{si}) becomes
\[
 A\left(y_{2}+y_{3}\right) -\frac{B}{y_{2} y_{3}}+y_{2} y_{3}\left(y_{1} y_{2}- y_{1}  y_{4}+ y_{3} y_{4}\right),
\]
whereas the mapping (\ref{mg}) reduces to
\begin{equation} \label{del}
\left(y_1,y_2,y_3,y_4\right)\mapsto \left(y_2,y_3,y_4,-\frac{ y_{1} y_{2}^{2} y_{3}^{2} y_{4}+A y_{2} y_{3} y_{4}+B }{ y_{2} y_{3}^{2} y_{4}^{2}}\right).
\end{equation}
This mapping has quadratic growth, but it does not possess the Laurent property. In fact, the mapping (\ref{mg}) is the Laurentification of (\ref{del}), cf. \cite{HK,HHVQ}.

\section{Periodic reduction of a higher Gel’fand-Dikii multilinear system}
How to algorithmically perform periodic reduction for systems of lattice equations in general is an open problem, cf. \cite{IVPs,SNV}. We provide the solution to this problem for the higher Gel’fand-Dikii (GD) multilinear system (\ref{eq:vfpsi}), which is defined on the stencils given in Figure \ref{stcs}.

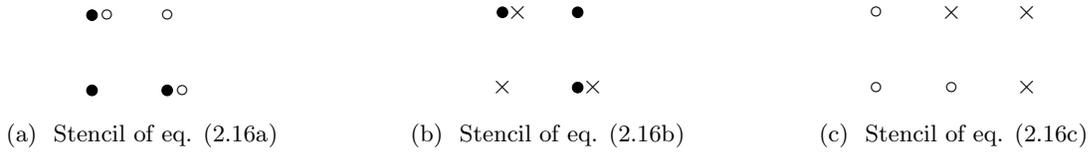
\begin{figure}[h!]
\begin{subfigure}{0.3\textwidth}
\begin{center}
\begin{tikzpicture} 
\filldraw [black] (0,0) circle (2pt)
(1,0) circle (2pt)
(0,1) circle (2pt);
\draw [black](1.2,0) node {$\circ$};
\draw [black](.2,1) node {$\circ$};
\draw [black](1,1) node {$\circ$};
\end{tikzpicture}
\caption{\label{st1} Stencil of eq. \eqref{eq:vfpsia}}
\end{center}
\end{subfigure}
\begin{subfigure}{0.3\textwidth}
\begin{center}
\begin{tikzpicture} 
\filldraw [black] (1,1) circle (2pt)
(1,0) circle (2pt)
(0,1) circle (2pt);
\draw [black](1.2,0) node {$\times$};
\draw [black](.2,1) node {$\times$};
\draw [black](0,0) node {$\times$};
\end{tikzpicture}
\caption{\label{st2} Stencil of eq. \eqref{eq:vfpsib}}
\end{center}
\end{subfigure}
\begin{subfigure}{0.3\textwidth}
\begin{center}
\begin{tikzpicture} 
\draw [black] (0,0) node {$\circ$}
(1,0) node {$\circ$}
(0,1) node {$\circ$};
\draw [black](2,0) node {$\times$};
\draw [black](1,1) node {$\times$};
\draw [black](2,1) node {$\times$};
\end{tikzpicture}
\caption{\label{st3} Stencil of eq. \eqref{eq:vfpsic}}
\end{center}
\end{subfigure}
\caption{\label{stcs} Stencils of the higher GD system, $\tau,\varphi,\psi$ are represented by $\bullet,\circ,\times$.}
\end{figure}

\begin{proposition} \label{wpivp}
A well-posed initial value (Cauchy) problem (IVP) with direction $\b{s}>\b{0}$ can be obtained by squeezing
the combined stencil in Figure \ref{csa}. For directions $\b{s}$ with $s_1>0,s_2<0$, a well-posed IVP can be obtained by squeezing
the combined stencil in Figure \ref{csb}.

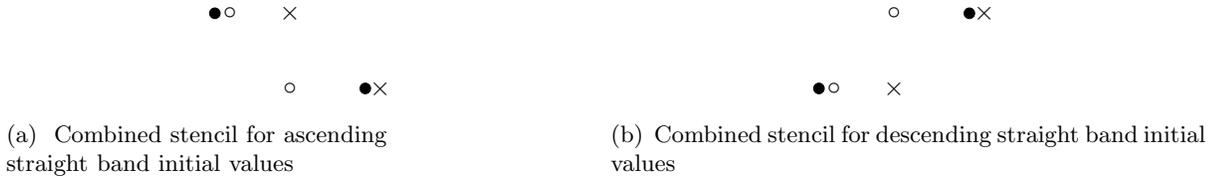
\begin{figure}[h!]
\begin{subfigure}{0.45\textwidth}
\begin{center}
\begin{tikzpicture} 
\filldraw [black] (0,1) circle (2pt)
                (2,0) circle (2pt);
\draw [black]   (0.2,1) node {$\circ$}
                (1,0) node {$\circ$};
\draw [black]   (1,1) node {$\times$}
                (2.2,0) node {$\times$};
\end{tikzpicture}
\caption{\label{csa} Combined stencil for ascending\\ straight band initial values}
\end{center}
\end{subfigure}
\begin{subfigure}{0.45\textwidth}
\begin{center}
\begin{tikzpicture} 
\filldraw [black] (0,0) circle (2pt)
                (2,1) circle (2pt);
\draw [black]   (0.2,0) node {$\circ$}
                (1,1) node {$\circ$};
\draw [black]   (1,0) node {$\times$}
                (2.2,1) node {$\times$};
\end{tikzpicture}
\caption{\label{csb} Combined stencil for descending straight band initial values}
\end{center}
\end{subfigure}
\caption{\label{cs} Combined stencils which can be used to construct Cauchy problems for the higher GD system}
\end{figure}
\end{proposition}

\begin{proof}
As in the scalar case, to obtain straight band initial values, the combined stencil is squeezed by lines with direction $\b{s}$ in the allowed region. Here we obtain three different bands, in which the initial values are given, one band for each variable. We need to show, for each combined stencil, that one can determine the value of each variable on the boundaries of its band, using one of the equations, in each direction. We start with the combined stencil  in \ref{csa}, solving for the variables on the boundaries on the right.  In order, one first solves eq. (\ref{eq:vfpsia}) for $\varphi$, then eq. (\ref{eq:vfpsic}) for $\psi$, and finally eq. (\ref{eq:vfpsib}) for $\tau$. To find the variables on the left boundaries, we first solve eq. (\ref{eq:vfpsib}) for $\psi$, then eq. (\ref{eq:vfpsic}) for $\varphi$, and then eq. (\ref{eq:vfpsia}) for $\tau$.  Next, we consider the combined stencil in \ref{csb} for descending bands. To determine the values of the variables on the boundaries on the right, one first solves eq. (\ref{eq:vfpsia}) for $\varphi$, then eq. (\ref{eq:vfpsib}) for $\tau$, and finally eq. (\ref{eq:vfpsic}) for $\psi$. To find the values on the left boundaries, we first solve eq. (\ref{eq:vfpsia}) for $\tau$, then eq. (\ref{eq:vfpsib}) for $\psi$, and then eq. (\ref{eq:vfpsic}) for $\varphi$.   \end{proof}

Let us use Proposition \ref{wpivp} to construct a well-posed IVP for the higher GD system (\ref{eq:vfpsi}). In Figure \ref{scsa} we squeeze the combined stencil with lines that have direction $\b{s}=(2,1)$.

\begin{figure}[h!]
\begin{subfigure}{0.45\textwidth}
\begin{center}
\begin{tikzpicture}[scale=1.5]  
\foreach \x in {-1,...,3}{
      \foreach \y in {-1,...,2}{
        \node[draw,circle,inner sep=0pt,fill] at (\x,\y) {};
      }
    }
\filldraw [white] (1,0) circle (1.4pt);
\filldraw [black] (0,1) circle (1.4pt)
                (2,0) circle (1.4pt);
\draw [black]   (1,1) node {$\times$}
                (2.1,0) node {$\times$};
\draw   (-1,1/2)--(2,2)
        (0,-1)--(3,1/2);
\draw[dashed]   (-.9,1/2)--(2.1,2)
        (-1,-1)--(3,1);
\draw[dotted]   (-1,0)--(3,2)
        (.1,-1)--(3,.45);
\filldraw [white] (0.1,1) circle (1.4pt)
            (1,0) circle (1.4pt);
\draw [black] (0.1,1) circle (1.4pt)
            (1,0) circle (1.4pt);      
\end{tikzpicture}
\caption{\label{scsa} Squeezing the combined stencil from Figure \ref{csa} \vspace{1.15cm}}
\end{center}
\end{subfigure}
\begin{subfigure}{0.45\textwidth}
\begin{center}
\begin{tikzpicture}[scale=1.5] 
\foreach \x in {-1,...,3}{
      \foreach \y in {-1,...,2}{
        \node[draw,circle,inner sep=0pt,fill] at (\x,\y) {};
      }
    }
\filldraw [white] (1,0) circle (1.4pt);
\filldraw [black] (0,1) circle (1.4pt)
                (2,0) circle (1.4pt)
                (0.9,1) circle (1.4pt)
                (-1.1,0) circle (1.4pt)
                (-.1,0) circle (1.4pt)
                (0,-1) circle (1.4pt)
                (-1.1,-1) circle (1.4pt)
                (.9,0) circle (1.4pt)
                (2.9,2) circle (1.4pt)
                (2.9,1) circle (1.4pt)
                (1.9,1) circle (1.4pt)
                (2,2) circle (1.4pt)
                (2,2) circle (1.4pt)
                ;
\draw [black]   (1.1,1) node {$\times$}
                (1.1,0) node {$\times$}
                (-.9,0) node {$\times$}
                (3.1,2) node {$\times$}
                (0.1,0) node {$\times$}
                (2.1,1) node {$\times$}
                (-.9,-1) node {$\times$}
                (3.1,1) node {$\times$}
                (2.1,0) node {$\times$}
                (0.1,-1) node {$\times$};
\draw [black]   (-1,-1.15) node {?}
                (1,-.15) node {?}
                (3,.85) node {?}
                (0,-1.15) node {?}
                (2,-.15) node {?}
                (0.1,-1.15) node {?}
                (2.1,-.15) node {?}
                ;
\draw [black]   (0,1.2) node {0}
                (2,2.2) node {0}
                (1,1.2) node {1}
                (3,2.2) node {1}
                (-1,0.2) node {1}
                (2,1.2) node {2}
                (0,0.2) node {2}
                (1,.2) node {3}
                (-1,-.8) node {3}
                (3,1.2) node {3}
                (2,.2) node {4}
                (0,-.8) node {4}
                ;
\draw   (-1,1/2)--(2,2)
        (0,-1)--(3,1/2);
\draw[dashed]   (-.9,1/2)--(2.1,2)
        (-1,-1)--(3,1);
\draw[dotted]   (-1,0)--(3,2)
        (.1,-1)--(3,.45);
\filldraw [white] (0.1,1) circle (1.4pt)
            (1,0) circle (1.4pt)
            (2.1,2) circle (1.4pt)
            (-1,0) circle (1.4pt)
            (1,1) circle (1.4pt)
            (3,2) circle (1.4pt)
            (0,0) circle (1.4pt)
            (2,1) circle (1.4pt)
            (-1,-1) circle (1.4pt)
            (3,1) circle (1.4pt);
\draw [black] (0.1,1) circle (1.4pt)
            (1,0) circle (1.4pt)
            (2.1,2) circle (1.4pt)
            (-1,0) circle (1.4pt)
            (1,1) circle (1.4pt)
            (3,2) circle (1.4pt)
            (0,0) circle (1.4pt)
            (2,1) circle (1.4pt)
            (-1,-1) circle (1.4pt)
            (3,1) circle (1.4pt);
\end{tikzpicture}
\caption{\label{scsb} Labelling the initial values, for evolution to the right. Values on the same line have the same label. Values with a question mark are initial values, they need to be calculated.}
\end{center}
\end{subfigure}
\end{figure}

The initial value configuration in Figure \ref{scsb} leads to the following 10-dimensional mapping
\begin{equation} \label{map}
\left(x_0,x_1,x_2,x_3,y_0,y_1,y_2,z_1,z_2,z_3\right)
\mapsto
\left(x_1,x_2,x_3,x_4,y_1,y_2,y_3,z_2,z_3,z_4\right), \end{equation}
where $y_3$ is determined by taking $n=0$ in
\begin{subequations} 
\begin{equation} \label{eA}
a x_{n+2} y_{n+1}+b x_{n} y_{n+3}=\left(a +b \right) x_{n+3} y_{n},
\end{equation}
after which $z_4$ can be found from, with $n=0$, 
\begin{equation} \label{eC}
a y_{n+2} z_{n+2}+d y_{n} z_{n+4}=\left(a +d \right) y_{n+3} z_{n+1},
\end{equation}
and then $x_4$ is the solution of,  with $n=0$,
\begin{equation} \label{eB}
a x_{n+2} z_{n+3}+c x_{n+4} z_{n+1}=\left(a +c \right) x_{n+1} z_{n+4}.
\end{equation}
\end{subequations}
Here $x,y,z$ denote the reductions of $\tau,\varphi,\psi$, and
\[
a=p-q, b=q-\omega_1(p), c=q-\omega_2(p), d=q-\omega_3(p).
\]
The inverse of the map \eqref{map} is given by $\left(x_0,x_1,x_2,x_3,y_0,y_1,y_2,z_1,z_2,z_3\right)
\mapsto
\left(x_{-1},x_0,x_1,x_2,y_{-1},y_1,y_2,z_0,z_1,z_2\right)$, where we obtain $z_0$ from \eqref{eB} ($n=-1$), $y_{-1}$ from \eqref{eC} ($n=-1$), and $x_{-1}$ from \eqref{eA} ($n=-1$), in that order.  
\begin{proposition}
The mapping (\ref{map}) has the Laurent property.
\end{proposition}
\begin{proof}
To aid the proof we shift the index on the $y$-variable, so that the map becomes
\[
\mathfrak{m}:\left(x_0,x_1,x_2,x_3,y_1,y_2,y_3,z_1,z_2,z_3\right)
\mapsto
\left(x_1,x_2,x_3,x_4,y_2,y_3,y_4,z_2,z_3,z_4\right),
\]
with $a x_{2} y_{2}+b x_{0} y_{4}=\left(a +b \right) x_{3} y_{1}$, $a x_{2} z_{3}+c x_{4} z_{1}=\left(a +c \right) x_{1} z_{4}$, $a y_{3} z_{2}+d y_{1} z_{4}=\left(a +d \right) y_{4} z_{1}$. We define, for all $i$,
\[
\mathfrak{m}^i\left(x_0,x_1,x_2,x_3,y_1,y_2,y_3,z_1,z_2,z_3\right)=\left(x_i,x_{i+1},x_{i+2},x_{i+3},y_{i+1},y_{i+2},y_{i+3},z_{i+1},z_{i+2},z_{i+3}\right)\]
and $x_i=\dfrac{p^x_i}{q^x_i}$,
$y_i=\dfrac{p^y_i}{q^y_i}$,
$z_i=\dfrac{p^z_i}{q^z_i}$, with $\gcd(p^x_i,q^x_i)=\gcd(p^y_i,q^y_i)=\gcd(p^z_i,q^z_i)=1$. The first few numerators are
\begin{align*}
p^y_4&=-a x_{2} y_{2}+\left(a +b \right) x_{3} y_{1}\\
p^z_4&=-a b x_{0} y_{3} z_{2}-a \left(a +d \right) x_{2} y_{2} z_{1}+\left(a +d \right) \left(a +b \right) x_{3} y_{1} z_{1}\\
p^x_4&=-a b \left(a +c \right) x_{0} x_{1} y_{3} z_{2}-a b d x_{0} x_{2} y_{1} z_{3}-a \left(a +d \right) \left(a +c \right) x_{1} x_{2} y_{2} z_{1}+\left(a +d \right) \left(a +c \right) \left(a +b \right) x_{1} x_{3} y_{1} z_{1}\\
p^y_5&=-a b \left(a +c \right) \left(a +b \right) x_{0} x_{1} y_{2} y_{3} z_{2}-a b d \left(a +b \right) x_{0} x_{2} y_{1} y_{2} z_{3}-a b c d x_{0} x_{3} y_{1} y_{3} z_{1}-a \left(a +d \right) \left(a +c \right) \left(a +b \right) x_{1} x_{2} y_{2}^{2} z_{1}\\
&\ \ \ +\left(a +b \right)^{2} \left(a +d \right) \left(a +c \right) x_{1} x_{3} y_{1} y_{2} z_{1}\\
p^z_{5} &= 
-a b \left(a +d \right) \left(a +c \right) \left(a +b \right) x_{0} x_{1} y_{2} y_{3} z_{2}^{2}-a b d \left(a +d \right) \left(a +b \right) x_{0} x_{2} y_{1} y_{2} z_{2} z_{3}-a b c d \left(a +d \right) x_{0} x_{3} y_{1} y_{3} z_{1} z_{2}\\
&\ \ \ +a^{2} b c d x_{1} x_{2} y_{1} y_{2} z_{1} z_{3}-a \left(a +d \right)^{2} \left(a +c \right) \left(a +b \right) x_{1} x_{2} y_{2}^{2} z_{1} z_{2}-a b c d \left(a +b \right) x_{1} x_{3} y_{1}^{2} z_{1} z_{3} \\
&\ \ \ +\left(a +d \right)^{2} \left(a +b \right)^{2} \left(a +c \right) x_{1} x_{3} y_{1} y_{2} z_{1} z_{2}\\
p^x_{5} &= 
-a b \left(a +c \right)^{2} \left(a +d \right) \left(a +b \right) x_{0} x_{1} x_{2} y_{2} y_{3} z_{2}^{2}+a^{2} b^{2} c d y_{3} x_{0} x_{1} x_{3} y_{2} z_{1} z_{2}-a b d \left(a +d \right) \left(a +c \right) \left(a +b \right) x_{0} x_{2}^{2} y_{1} y_{2} z_{2} z_{3} \\
&\ \ \ -a b c d \left(a +d \right) \left(a +c \right) x_{0} x_{2} x_{3} y_{1} y_{3} z_{1} z_{2}+a^{2} b c d \left(a +c \right) x_{1} x_{2}^{2} y_{1} y_{2} z_{1} z_{3}-a \left(a +d \right)^{2} \left(a +c \right)^{2} \left(a +b \right) x_{1} x_{2}^{2} y_{2}^{2} z_{1} z_{2}\\
&\ \ \ -a b c d \left(a +c \right) \left(a +b \right) x_{1} x_{2} x_{3} y_{1}^{2} z_{1} z_{3}+\left(a +d \right)^{2} \left(a +c \right)^{2} \left(a +b \right)^{2} x_{1} x_{2} x_{3} y_{1} y_{2} z_{1} z_{2}+a^{2} b c d \left(a +d \right) x_{1} x_{2} x_{3} y_{2}^{2} z_{1}^{2}\\
&\ \ \ -a b c d \left(a +d \right) \left(a +b \right) x_{1} x_{3}^{2} y_{1} y_{2} z_{1}^{2}\\
& \ \ \vdots
\end{align*}
The denominators of the first five iterates are:
\begin{align*}
q^y_4&=b x_{0},\ 
q^y_5=b^{2} c d  x_{0} y_{1} z_{1} x_{1},\
q^y_6=b^{3}c^{2}  d^{2} x_{0}^{2} y_{1} z_{1} x_{1} y_{2} z_{2} x_{2},\
q^y_7=b^{5} c^{3} d^{3} x_{0}^{3} y_{1}^{2} z_{1}^{2} x_{1}^{2} y_{2} z_{2} x_{2} y_{3} z_{3} x_{3},\\
q^y_8&=b^{7} c^{5} d^{5} z_{1}^{3} y_{1}^{3} x_{0}^{3}x_{1}^{3} y_{2}^{2} z_{2}^{2} x_{2}^{2} y_{3} z_{3} x_{3},
q^z_4=b d x_{0} y_{1},
q^z_5=b^{2} c d^{2} x_{0} y_{1} z_{1} x_{1} y_{2},
q^z_6=b^{3} c^{2} d^{3} x_{0}^{2} y_{1}^{2} z_{1} x_{1} y_{2} z_{2} x_{2} y_{3},\\
q^z_7&=b^{5} c^{3} d^{5} x_{0}^{3} x_{3} z_{3} y_{3} x_{2} z_{2} y_{2}^{2} x_{1}^{2} z_{1}^{2} y_{1}^{3},
q^z_8=b^{7} c^{5} d^{7} 
x_{1}^{3} z_{1}^{3} y_{1}^{3} x_{0}^{4} y_{2}^{3} z_{2}^{2} x_{2}^{2} y_{3}^{2} z_{3} x_{3},
q^x_4=b c d x_{0} y_{1} z_{1},
q^x_5=b^{2} c^{2} d^{2} x_{0} y_{1} z_{1} x_{1} y_{2} z_{2},\\
q^x_6&=b^{3} c^{3} d^{3} x_{0}^{2} y_{1}^{2} z_{1}^{2} x_{1} y_{2} z_{2} x_{2} y_{3} z_{3},
q^x_7=b^{5}  c^{5} d^{5} y_{1}^{3} x_{0}^{3} x_{1}^{2} y_{2}^{2} z_{2}^{2} x_{2} y_{3} z_{3} x_{3} z_{1}^{3},
q^x_8=b^{7} c^{7} d^{7} y_{2}^{3} x_{1}^{3} z_{1}^{3} y_{1}^{4} x_{0}^{4} z_{2}^{3} x_{2}^{2} y_{3}^{2} z_{3}^{2} x_{3}. 
\end{align*}
We have verified, using MAPLE \cite{MPL}, that, for $i=5,\ldots,8$, $\gcd(x^p_4,x^p_i)=\gcd(x^p_4,y^p_i)=\gcd(x^p_4,z^p_i)=\gcd(y^p_4,x^p_i)=\gcd(y^p_4,y^p_i)=\gcd(y^p_4,z^p_i)=\gcd(z^p_4,x^p_i)=\gcd(z^p_4,y^p_i)=\gcd(z^p_4,z^p_i)=1$. From this, and the fact that the denominators $x^q_i,y^q_i,z^q_i$, $i=4,\ldots,8$ are monomial, by \cite[Theorem 2]{HHVQ}, for all $i$, the functions $x_i,y_i,z_i$ are Laurent polynomials in the initial values.
\end{proof}

Notice that the denominators do not depend on $a$. Hence, for any $a\in\mathbb{Z}$ and with all initial values equal to $\pm1$, the sequences $x,y,z$ are integer sequences. Consider the map \eqref{map} and its inverse (which also possesess the Laurent property). Taking $a=-2$, $z_1=-1$, and setting all other initial values to $1$, the values for $x_n,y_n,z_n$ with $-6\leq n \leq 12$ are respectively
\begin{align*}
67177, 4209, -255, 137, -15, 5,\ &1, 1, 1, 1, 1, 5, -7, 49, -767, 2905, -5519, 1021381, 22876241\\
22417, 169, 129, -47, 9, -3,\ &1, 1, 1, 1, 1, -3, 17, -7, -127, 5521, -15847, -89667, -54620095\\
-74567, 1079, -631, 73, -5, -7,\ &1, -1, 1, 1, 3, 1, 9, -57, 313, -647, 13835, 312009, 3909457.
\end{align*}
This provides an example of a set of three integer sequences generated by a coupled system of bilinear recurrences. These sequences satisfy many more relations. For example, from \eqref{eq:vfvf}, we get
\begin{align*}
&a^{2} \left(x_{n -2} x_{n -1} y_{n} y_{n -3}-x_{n} x_{n -3} y_{n -2} y_{n -1}\right)+\left(a +d \right) \left(a +c \right) \left(x_{n} x_{n -4} y_{n} y_{n -2}-x_{n -3} x_{n -1} y_{n -3} y_{n +1}\right)\\
&+c d \left(x_{n -3} x_{n +2} y_{n -3} y_{n -2}-x_{n} x_{n -1} y_{n} y_{n -5}\right)=0,
\end{align*}
and \eqref{eq:hexaGD} yields
\begin{align*}
&a^{3} \left(x_{n}^{2} x_{n +1}^{2} x_{n -2} x_{n -3}-x_{n -1}^{2} x_{n +2} x_{n +1} x_{n -2}^{2}\right)-\left(a +d \right) \left(a +c \right) \left(a +b \right) x_{n +1} x_{n -2} \left(x_{n +2} x_{n}^{2} x_{n -4}-x_{n -1}^{2} x_{n +3} x_{n -3}\right)\\
&-b c d \left(x_{n -1} x_{n -2} x_{n -3} \left(x_{n} x_{n +4} x_{n -1}+x_{n +3} x_{n -2} x_{n +2}\right)-x_{n +2} x_{n +1} x_{n} \left(x_{n} x_{n -5} x_{n -1}+x_{n -3} x_{n +1} x_{n -4}\right)\right)=0,
\end{align*}
which also holds for the $y$ and $z$ sequences.
\section*{Acknowledgement}
The paper was initiated while FWN was an affiliate of the Sydney Mathematical Research Institute (SMRI), and he is grateful for its support. He also would like to thank the Department of Mathematical and Physical Sciences of La Trobe University for its hospitality during this visit. He was supported by EPSRC grant EP/007290/1 when most of the work was done. 
We are also grateful to Michael Somos for pointing out that the sequence \eqref{eq:Somseq} can 
be obtained from the bilinear recursion \eqref{eq:Somrecurs}.

\end{document}